\documentclass[AMA,STIX1COL]{WileyNJD-v2}
\usepackage{moreverb}
\usepackage[outdir=./]{epstopdf}

\articletype{RESEARCH ARTICLE}%

\begin{document}

\title{Stabilizing of a Class of Underactuated Euler Lagrange System Using an Approximate Model}

\author[1]{Huseyin Alpaslan Yildiz*}

\author[2]{Leyla Goren-Sumer}

\authormark{YILDIZ AND GOREN-SUMER}

\address[1]{\orgdiv{Siemens A.S}, \orgname{ADV D EU TR}, \orgaddress{\state{Istanbul}, \country{Turkey}}}

\address[2]{\orgdiv{Istanbul Technical University}, \orgname{Dept. of Control and Automation Eng.}, \orgaddress{\state{Istanbul}, \country{Turkey}} \email{leyla.goren@itu.edu.tr}}

\corres{*Huseyin Alpaslan Yildiz, Siemens Sanayi ve Tic. A.S. Yakacik Cad. No:111 34870 Istanbul, Turkey. \email{alpaslan.yildiz@siemens.com}}


\abstract[Abstract]{ 
The energy shaping method, Controlled Lagrangian, is a well-known approach to stabilize the under-actuated Euler Lagrange (EL) systems. In this approach, to construct a control rule, some nonlinear, nonhomogeneous partial differential equations (PDEs), which are called matching conditions, must be solved. In this paper, a method is proposed to obtain an approximate solution of these matching conditions for a class of under-actuated EL systems. To develop the method, the potential energy matching condition is transformed to a set of linear PDEs using an approximation of inertia matrices.  So the assignable potential energy function and the controlled inertia matrix, both are constructed as a common solution of these PDEs. Afterwards, the gyroscopic and dissipative forces are found as the solution of the kinetic energy matching condition. Finally, the control rule is constructed by adding energy shaping rule and additional dissipation injection to provide asymptotic stability. The stability analysis of the closed loop system which used the control rule derived with the proposed method is also given. To demonstrate the success of the proposed method, the stability problem of the inverted pendulum on a cart is considered.
}


\maketitle

\section{Introduction}\label{sec1}

The stabilization problem of nonlinear systems at a desired equilibrium point is an attractive subject for control researchers ever since the tools are developed to analyze nonlinear systems. When the energy shaping and dissipation injection control policy is used, the mathematical structure of the system is preserved, but these approaches are versatile just for a special class of nonlinear systems where Euler Lagrange and port-controlled Hamiltonian systems (PCH). The energy shaping control methodology is to render the open-loop system to a closed loop system, which has a stable desired equilibrium point, via feedback. As long as the number of control input is equal to the number of degrees of freedom (DoF) of the open-loop system, the system is called full-actuated system and shaping potential energy is adequate to stabilize the closed-loop system at desired equilibrium. However under-actuated systems can be stabilize not only by modifying the potential energy of the system but also kinetic energy of the system. This idea is introduced first in Ailon and Ortega\cite{Ailon1993} and it is called total energy shaping.  In literature, there are two main approaches to shape the total energy functions of systems which are developed separately for Euler Lagrange and Hamiltonian systems, entitled Controlled Lagrangians\cite{Bloch2000} and Interconnection and damping assignment passivity-based control (IDA-PBC)\cite{Ortega1998}, respectively.

In these methods, the existence conditions of a feedback law, which stabilizes the system at a desired equilibrium point, are given as a set of nonlinear partial differential equations (PDEs) known as matching conditions. Several studies have been dedicated to derive the solution of these matching equations, as in Gomez et al\cite{Gomez2001}, Acosta et al\cite{Acosta2005} also Viola et al\cite{Viola2007} and references therein. Additionally, some methods are also proposed in Hamberg\cite{Hamberg1999}, Bloch et al\cite{Bloch2000,Bloch2001}, Ortega et al\cite{Ortega2002}, Chang\cite{Chang2005}, Auckly and Kapitansky\cite{Auckly2006}. As recognized from these works, there are lots of difficulties about solving related PDEs, so the stabilization problem is still known as a hard problem, for the underactuated case. To solve this problem, in Goren-Sumer and Yalçin\cite{Sumer2011}, it is shown that the discrete time formulation of the problem is possible and then a relatively easy technique to solve these matching equations that are in the form of PDEs is given.   Some other approaches to the same problem are proposed, i.e., in Sarras at al\cite{Sarras2013}, where a method is based on the immersion and invariance methodology and in Donaire et al\cite{Donaire2016}, it is shown that the solution of partial differential equations cannot be needed whenever the system satisfies some assumptions.

To obtain an approximate solution of these matching equations, a method based on the constructing the approximate model of the open-loop and closed-loop systems is proposed in Goren-Sumer and Sengor\cite{Sumer2015}. This approach makes it possible to derive a new matching conditions in the form of a set of linear partial differential equations and a set of linear equations to solve the stabilization problem of underactuated systems. In this study, this approach is modified and expanded to a class of underactuated EL systems. To the best of our knowledge, no similar approach is available in the literature.

\section{preliminaries}\label{sec2}

Consider an EL system in n-dimensional configuration space Q\cite{Ortega1998}:
\begin{eqnarray} \label{eq:1}
	\frac{d}{d t} \frac{\partial d}{\partial \dot{{q}}} L({q}, \dot{{q}})-\frac{\partial d}{\partial {q}} L({q}, \dot{{q}})=G({q}) u({q}, \dot{{q}})
\end{eqnarray}
In (\ref{eq:1}) if $G({q}) \in \mathbb{R}^{n \times m}$ and $\operatorname{rank} G({q})<n$, the system is called underactuated.

Let us define a desired closed loop EL system as:
\begin{eqnarray} \label{eq:2}
	\frac{d}{d t} \frac{\partial d}{\partial \dot{{q}}} L_{c}({q}, \dot{{q}})-\frac{\partial d}{\partial {q}} L_{c}({q}, \dot{{q}})+F_{c}({q}, \dot{{q}})=0
\end{eqnarray}
in which $F_{c}({q}, {\dot{q}})$ is defined as gyroscopic and/or dissipation forces \cite{Chang2002}.

In order to solve the stabilization problem for EL systems, the controlled Lagrangian design method was developed in Bloch et al\cite{Bloch2000,Bloch2001}.  The main idea in this method was to design a stabilizing controller is based on assigning a desired Lagrangian Function for control,
\begin{eqnarray} \label{eq:3}
	L_{c}({q}, {\dot{q}})=\frac{1}{2} {\dot{q}^{T}} M_{c}({q}) {\dot{q}}-V_{c}({q})
\end{eqnarray}
such that $V_{c}({q})$ has an isolated minimum point which coincides with the desired equilibrium point $({q}^{*}, \mathbf{0})$ of the closed loop system. This problem can be solved only by assigning desired potential energy function for full actuated EL systems. To solve this problem, first it is required to determine the desired closed-loop Lagrangian function $L_{c}({q}, \dot{{q}})$, namely $M_{c}({q})$ and $V_{c}({q})$, after than,  a gyroscopic force  and/or dissipation force , $F_{c}({q}, {\dot{q}})$ must be found\cite{Chang2002}.

For a simple mechanical system, a motion equation for open loop EL and desired closed loop EL system is defined respectively as follows:
\begin{eqnarray} \label{eq:4}
	M \ddot{{q}}+\frac{\partial M \dot{{q}}}{\partial {q}} \dot{{q}}-\frac{1}{2} \frac{\partial \dot{{q}}^{T} M \dot{{q}}}{\partial {q}}+\frac{\partial V}{\partial {q}}=G {u}
\end{eqnarray}
\begin{eqnarray} \label{eq:5}
	M_{c} \ddot{{q}}+\frac{\partial M_{c} \dot{{q}}}{\partial {q}} \dot{{q}}-\frac{1}{2} \frac{\partial \dot{{q}}^{T} M_{c} \dot{{q}}}{\partial {q}}+\frac{\partial V_{c}}{\partial {q}}+{F}_{c}=0
\end{eqnarray}
and $C({q}, \dot{{q}})$ which is "Coriolis and centrifugal forces" matrix\cite{Ortega1998} is defined as,
\begin{eqnarray} \label{eq:6}
	C({q}, \dot{{q}})=\frac{\partial M \dot{{q}}}{\partial {q}}-\frac{1}{2}\left(\frac{\partial M \dot{{q}}}{\partial {q}}\right)^{T}
\end{eqnarray}
The existence condition of feedback law ${u}({q}, \dot{{q}})$ which transforms the system (\ref{eq:4}) to the closed loop system (\ref{eq:5}) is given by the matching conditions and the procedure to obtain these are summarized as follows\cite{Blankenstein2002}:
\begin{eqnarray} \label{eq:7}
	\ddot{{q}}=M^{-1} G {u}-M^{-1} \frac{\partial M \dot{{q}}}{\partial {q}} \dot{{q}}+\frac{1}{2} M^{-1} \frac{\partial \dot{{q}}^{T} M \dot{{q}}}{\partial {q}}-M^{-1} \frac{\partial V}{\partial {q}}
\end{eqnarray}
\begin{eqnarray} \label{eq:8}
	\ddot{{q}}=-M_{c}^{-1} \frac{\partial M_{c} \dot{{q}}}{\partial {q}} \dot{{q}}+\frac{1}{2} M_{c}^{-1} \frac{\partial \dot{{q}}^{T} M_{c} \dot{{q}}}{\partial {q}}-M_{c}^{-1} \frac{\partial V_{c}}{\partial {q}}-M_{c}^{-1} F_{c}
\end{eqnarray}
where ${F}_{{c}}=(J+R) \dot{{q}}$ and $$J=-J^{T}, R \geq 0$$(see after the second paragraph of Definition 2.1 in Chang\cite{Chang2002}).  

Then the matching conditions, as in form of nonlinear PDEs, can be written as follows,
\begin{eqnarray} \label{eq:9}
	G^{\perp}\left(\frac{\partial V}{\partial {q}}-M M_{c}^{-1} \frac{\partial V_{c}}{\partial {q}}\right)=0
\end{eqnarray}
\begin{eqnarray} \label{eq:10}
	G^{\perp}\left[\left(\left(\frac{\partial M \dot{{q}}}{\partial {q}}-\frac{1}{2}\left(\frac{\partial M \dot{{q}}}{\partial {q}}\right)^{T}\right)-M M_{c}^{-1}\left(\frac{\partial M_{c} \dot{{q}}}{\partial {q}}-\frac{1}{2}\left(\frac{\partial M_{c} \dot{{q}}}{\partial {q}}\right)^{T}\right)-M M_{c}^{-1}(J+R)\right) \dot{{q}}\right]=0
\end{eqnarray}
where $G^{\perp} :\left(\mathbb{R}^{n-m}\right)^{T} \rightarrow\left(\mathbb{R}^{n}\right)^{T}$ is left annihilator of $G$. The equation (\ref{eq:9}) and (\ref{eq:10}) are called as the potential energy matching condition and the kinetic energy matching condition, respectively. To construct the controller, first the matrices $M_{c}({q})>0, J=-J^{T}$ and $R \geq 0$, which hold the matching conditions (\ref{eq:10}) must be found.  After that the desired potential energy function $V_{c}({q})$ which has an isolated minimum at $\left({q}^{*}, 0\right)$ should be found as a solution of (\ref{eq:9}). Energy shaping feedback rule can be obtained as, 
\begin{eqnarray} \label{eq:11}
	\begin{gathered}
	u_{c}=\left(G^{T} G\right)^{-1} G^{T}\left\{\left[\left(\frac{\partial M \dot{{q}}}{\partial {q}}-\frac{1}{2}\left(\frac{\partial M {q}}{\partial {q}}\right)^{T}\right)-M M_{c}^{-1}\left(\frac{\partial M_{c} \dot{{q}}}{\partial {q}}-\frac{1}{2}\left(\frac{\partial M_{c} \dot{{q}}}{\partial {q}}\right)^{T}\right)-M M_{c}^{-1}(J+R)\right] \dot{{q}}\right\}\\
	+\left(G^{T} G\right)^{-1} G^{T}\left\{\left[\frac{\partial V}{\partial {q}}-M M_{c}^{-1} \frac{\partial V_{c}}{\partial {q}}\right]\right\}
	\end{gathered}
\end{eqnarray}
using $M_{c}({q})>0$, $V_{c}({q})$, $J=-J^{T}$ and $R \geq 0$\cite{Blankenstein2010}. Finally, dissipation must be injected to the closed loop system to guarantee asymptotic stability of closed loop system which is given as Blankenstein et al \cite{Blankenstein2010} and Ortega and Garcia-Canseco\cite{Ortega2004},
\begin{eqnarray} \label{eq:12}
u_{d}=-K_{v} G^{T} M^{-1} M_{c} {\dot{q}}
\end{eqnarray}
where $K_{v}>0$.

\section{main results}\label{sec3}

Since the matching conditions given in (\ref{eq:9}), (\ref{eq:10}) are nonlinear and nonhomogeneous PDEs, they are too difficult to solve. Also, as mention in Blankenstein et al\cite{Blankenstein2010} and Viola et al\cite{Viola2007} there is no general solution for them. The main idea of this study is to propose a method that finds an approximate solution of PDEs given in (\ref{eq:9}) namely potential energy matching condition. To fulfill this, it is found that a potential energy function $V_{c}({q})$ has an isolated minimum point, which coincides with the desired equilibrium point $\left({q}^{*}, 0\right)$ of the closed loop system, and a controlled generalized inertia matrix $\widehat{M}_{c}({q})$ such that approximately satisfies the potential energy matching condition given in (\ref{eq:9}).

The method proposed in this paper might be summarized as follows:
\begin{step} \label{step1}
	Let us define $r$ sub-regions in the configuration space of the EL system, some scalar functions $h_i(q)$s,
	\begin{eqnarray} \label{eq:13}
	S_{i} \triangleq\left\{{q} | h_{i}({q}) \geq h_{l}({q}), \quad l=1,2, \ldots, r, \quad i \neq l\right\}
	\end{eqnarray}
	where $h_{i}({q})$s are defined as $0<h_{i}({q}) \leq 1$, also define $r$ number of points, $q_i$s, for each sub-regions such that $\left.h_{i}(q)\right|_{q=q_{i}}=1$.
	\begin{remark} \label{remark1}
		For instance, they may be chosen as radial basis functions. In this paper, the functions $h_{i}({q})$s are chosen as ${h_{i}({q})=e^{-\left(\epsilon_{i}\left\|{q}-{q}_{i}\right\|\right)^{2}}}$, in which the ${q}_{i}$s are the centers of the subregions.
	\end{remark} 
\end{step} 
\begin{step} \label{step2}
	Let us write the potential energy matching condition at ${q}={q}_{i}$s, thus $r$ number of equations are obtained as follows,
	\begin{eqnarray} \label{eq:14}
	G^{\perp}\left[\frac{\partial V}{\partial {q}}-M_{i} M_{c i}^{-1} \frac{\partial V_{c}}{\partial {q}}\right]=0, \qquad \forall i
	\end{eqnarray}
	in which $M_{i}=M\left({q}_{i}\right)$. Then find a controlled potential energy function $V_{c}({q})$ as a common solution of (\ref{eq:14}), such that $V_{c}({q})$ has the properties of $\partial V_{c}({q}) / \partial{q}|_{q=q^{*}}=0$  and $\partial^{2} V_{c}({q}) / \partial{q}^{2}>0$, and also this common solution makes it possible to find $r$ numbers of matrices $M_{c i}>0$ satisfying equation (\ref{eq:14}). For the existence of such solutions a lemma will be given later.
	\begin{remark} \label{remark2}
		The existence of the common solution of PDEs given in (\ref{eq:14}), namely the existence of $V_{c}({q})$ and $M_{c i}>0, \forall i$, for given $M_i$s, determine the class of Euler Lagrange systems which are stabilized via the method proposed here. The $M_{ci}>0$ which satisfied (\ref{eq:14}) do not need to be unique, in this case $M_{ci}$s are expressed in parametric from.
	\end{remark} 
\end{step} 
\begin{step} \label{step3}
	Let us define an generalized inertia matrix of closed loop system in terms of $h_i({q})$s defined in \textit{Step 1} and $M_{ci}$s found in \textit{Step 2},
	\begin{eqnarray} \label{eq:15}
	\widehat{M}_{c}({q})=\sum_{i=1}^{r} \left(h_{i}({q}) \widetilde{M}_{c i}+M_{c b i}\right) 
	\end{eqnarray}
	where $\widetilde{M}_{c i}+M_{c b i} = M_{ci}$. Find the parameters of $h_i\left({q}\right)$s, and constant matrices $\widetilde{M}_{c i}$ and $M_{c b i}$ for $i=1,2, \ldots, r,$ such that the following expression holds,
	\begin{eqnarray} \label{eq:16}
	\min \left\|G^{\perp}\left[\frac{\partial V}{\partial {q}}-M\widehat{M}_{c}^{-1}\left(\frac{\partial V_{c}}{\partial {q}}\right)\right]\right\|
	\end{eqnarray}
	\begin{remark} \label{remark3}
		If the $h_i ({q})$s are chosen as $h_{i}({q})=e^{-\left(\epsilon_{i}\left\|{q}-{q}_{i}\right\|\right)^{2}}$, the parameter $\epsilon_{i}$s and the bias terms are found. If the $M_{ci}$s have been expressed in parametric form mentioned in \textit{Remark 2}, then the proper values of these parameters can also be determined in this step.
	\end{remark}
\end{step} 
\begin{step} \label{step4}
	Construct an approximate generalized inertia matrix of closed loop system, $\widehat{M_c}(q)$ given in (\ref{eq:15}), using the scalar $h_i\left({q}\right)$s, and constant matrices $\widetilde{M}_{c i}$s and $M_{c b i}$s found in \textit{Step 3}.
\end{step} 
\begin{step} \label{step5}
	Find some $J({q}, \dot{{q}})=-J^{T}({q}, \dot{{q}})$ and $R({q}, \dot{{q}})=R^{T}({q}, \dot{{q}}) \geq 0$ matrices such that they hold the kinetic energy matching condition as follows,
	\begin{eqnarray} \label{eq:17}
	G^{\perp}\left[\left(\left(\frac{\partial M \dot{{q}}}{\partial {q}}-\frac{1}{2}\left(\frac{\partial M \dot{{q}}}{\partial {q}}\right)^{T}\right)-M \widehat{M}_{c}^{-1}\left(\frac{\partial \widehat{M}_{c} \dot{{q}}}{\partial {q}}-\frac{1}{2}\left(\frac{\partial \widehat{M}_{c} \dot{{q}}}{\partial {q}}\right)^{T}\right)-M \widehat{M}_{c}^{-1}(J+R)\right) \dot{{q}}\right]=0
	\end{eqnarray}
	For this aim let us consider $C({q}, \dot{{q}})$ and $\widehat{C}_{c}({q}, \dot{{q}})$ defined as,
	\begin{eqnarray} \label{eq:18}
		\begin{gathered}
		C({q}, \dot{{q}})=\frac{\partial M({q}) \dot{{q}}}{\partial {q}}-\frac{1}{2}\left(\frac{\partial M({q}) \dot{{q}}}{\partial {q}}\right)^{T} \\
		\widehat{C}_{c}({q}, \dot{{q}})=\frac{\partial \widehat{M}_{c}({q}) \dot{{q}}}{\partial {q}}-\frac{1}{2}\left(\frac{\partial \widehat{M}_{c}({q}) \dot{{q}}}{\partial {q}}\right)^{T}
		\end{gathered}
	\end{eqnarray}
	then kinetic energy matching condition given in (\ref{eq:17}) is rewritten as follows:
	\begin{eqnarray} \label{eq:19}
		G^{\perp} M \widehat{M}_{c}^{-1}\left[\widehat{M}_{c} M^{-1} C-\widehat{C}_{c}-(J+R)\right] \dot{q}=0
	\end{eqnarray}
	To obtain the solution of kinetic energy matching condition $J(q,\dot{q})$ can be taken as,
	\begin{eqnarray} \label{eq:20}
		J(q,\dot{q})=\frac{1}{2}\left[\left(\widehat{M}_{c} M^{-1} C-\widehat{C}_{c}\right)-\left(\widehat{M}_{c} M^{-1} C-\widehat{C}_{c}\right)^{T}\right]
	\end{eqnarray}
	and also any $R(q,\dot{q})\geq0$ need to be calculated which satisfies the condition below:
	\begin{eqnarray} \label{eq:21}
		G^{\perp}\left[\frac{1}{2}\left(\left(C-M \widehat{M}_{c}^{-1} \widehat{C}_{c}\right)+\left(C-M \widehat{M}_{c}^{-1} \widehat{C}_{c}\right)^{T}\right)-R\right]=0
	\end{eqnarray}
\end{step} 
\begin{step} \label{step6}
	Construct energy shaping control rule using the following relation,
	\begin{eqnarray} \label{eq:22}
	\begin{gathered}
	u_{c}=\left(G^{T} G\right)^{-1} G^{T}\left\{\left[\left(\frac{\partial M \dot{{q}}}{\partial {q}}-\frac{1}{2}\left(\frac{\partial M \dot{{q}}}{\partial {q}}\right)^{T}\right)-M \widehat{M}_{c}^{-1}\left(\frac{\partial \widehat{M}_{c} \dot{{q}}}{\partial {q}}-\frac{1}{2}\left(\frac{\partial \widehat{M}_{c} \dot{{q}}}{\partial {q}}\right)^{T}+(J+R)\right)\right] \dot{{q}}\right\} \\
	+\left(G^{T} G\right)^{-1} G^{T}\left\{\left[\frac{\partial V}{\partial {q}}-M \widehat{M}_{c}^{-1} \frac{\partial V_{c}}{\partial {q}}\right]\right\}
	\end{gathered}
	\end{eqnarray}
\end{step} 
\begin{step} \label{step7}
	Inject dissipation to satisfy system asymptotic stability where $K_v>0$,
	\begin{eqnarray} \label{eq:23}
	u_{d}=-K_{v} G^{T} M^{-1} \widehat{M}_{c} {\dot{q}}
	\end{eqnarray}
	therefore the control rule is constructed as follows,
	\begin{eqnarray} \label{eq:24}
	u=u_{c}+u_{d}
	\end{eqnarray}
\end{step}
The following lemma gives a sufficient condition for the existence of the common solution of (\ref{eq:14}) in \textit{Step 2}.
\begin{lemma} \label{lemma1}
	If there is $M_{c1}>0$ and $V_c({q})$ with the properties of $\nabla_{q} V_{c}\left({q}^{*}\right)=0$ and $V_{\text { chess }}^{{q}}\left({q}^{*}\right)>0$ which hold
	\begin{eqnarray} \label{eq:25}
	G^{\perp}\left(\frac{\partial V}{\partial {q}}-M_{1} M_{c 1}^{-1} \frac{\partial V_{c}}{\partial {q}}\right)=0
	\end{eqnarray}
	and there is $r$ number of matrices which the following equation hold:
	\begin{eqnarray} \label{eq:26}
	G^{\perp}\left(M_{i} M_{c i}^{-1}-M_{j} M_{c j}^{-1}\right)=0, \quad \forall i, j
	\end{eqnarray}
	than $V_c({q})$ is the common solution of the given PDE in (\ref{eq:14}), namely,
	$$G^{\perp}\left[\frac{\partial V}{\partial {q}}-M_{i} M_{c i}^{-1} \frac{\partial V_{c}}{\partial {q}}\right]=0, \qquad \forall i$$
\end{lemma}
\begin{proof}[Proof]
	Suppose satisfy the function $V_c({q})$ and $M_{ci}>0$ for $i=1$, the following relation,
	$$G^{\perp}\left(\frac{\partial V}{\partial {q}}-M_{1} M_{c 1}^{-1} \frac{\partial V_{c}}{\partial {q}}\right)=0$$
	Then, it can be written,
	$$G^{\perp} M_{i} M_{c i}^{-1}=G^{\perp} M_{j} M_{c j}^{-1}, \qquad \forall i, j$$
	so that one concludes that
	$$G^{\perp} \frac{\partial V}{\partial q}-\underbrace{G^{\perp} M_{1} M_{c 1}^{-1}}_{G^{\perp} M_{j} M_{c j}^{-1}} \frac{\partial V_{c}}{\partial q}=0, \quad \forall i, j$$
	They are nothing but only the PDEs given in (\ref{eq:14}), this proves the claim.
\end{proof}

\section{stability analysis}\label{sec4}

The standard controlled Lagrangian method guarantees the stability of the desired equilibrium over the assigned total energy function and dissipation injection. However, the method proposed here does not guarantee the stability of the system at the desired equilibrium for all choices of the number of $r$, namely  the amount of sub-regions. This is due to the fact that, the energy function of the closed loop system assigned via the control rule (\ref{eq:24}) cannot be known exactly, but only known approximately. Therefore, the existence condition of a number of the sub-regions which is guaranteed the stability must be determined. In this section, this issue will be handled by using the results given in Butz\cite{Butz1969}, Heinen\cite{Heinen1970} and Ahmadi\cite{Ahmadi2008,Ahmadi2014}.

To obtain the motion equation of the closed loop system let substitute $u({q}, \dot{{q}})$ (\ref{eq:24}) to the system equation (\ref{eq:4}),
\begin{eqnarray} \label{eq:27}
	\begin{gathered}
	M \ddot{{q}}+C \dot{{q}}+\frac{\partial V}{\partial {q}}= \\ G\left(G^{T} G\right)^{-1} G^{T}\left\{\left[C-M \widehat{M}_{c}^{-1} \widehat{C}_{c}-M \widehat{M}_{c}^{-1}(J+R)\right] \dot{{q}}+\left[\frac{\partial V}{\partial {q}}-M \widehat{M}_{c}^{-1} \frac{\partial V_{c}}{\partial {q}}\right]\right\}+G u_{d}
	\end{gathered}
\end{eqnarray}
and let consider the following relation,
\begin{eqnarray} \label{eq:28}
	G\left(G^{T} G\right)^{-1} G^{T}+\left(G^{\perp}\right)^{T}\left(G^{\perp}\left(G^{\perp}\right)^{T}\right)^{-1} G^{\perp}=I
\end{eqnarray}
where,
$$G_{n \times m}=\left[\begin{array}{c}{0} \\ {g_{m \times m}}\end{array}\right], \qquad G_{m \times n}^{\perp}=\left[\begin{array}{ll}{\tilde{g}_{p \times p}} & {0}\end{array}\right], \qquad p=(n-m)$$
thus the term $\left[I-\left(G^{\perp}\right)^{T}\left(G^{\perp}\left(G^{\perp}\right)^{T}\right)^{-1} G^{\perp}\right]$ can be used instead of $G\left(G^{T} G\right)^{-1} G^{T}$. To ease of the algebraic manipulations, assume that $g=I$, without loss of generality. Therefore, the terms of $\left[I-\left(G^{\perp}\right)^{T} G^{\perp}\right]$ and $G G^{T}$ can be used instead of $\left[I-\left(G^{\perp}\right)^{T}\left(G^{\perp}\left(G^{\perp}\right)^{T}\right)^{-1} G^{\perp}\right]$ and $G\left(G^{T} G\right)^{-1} G^{T}$, respectively, through to paper, since $\left(G^{\perp}\left(G^{\perp}\right)^{T}\right)^{-1}=I$ and $\left(G^{T} G\right)^{-1}=I$.

After some algebraic operations on (\ref{eq:27}), it is obtained as follows,  
\begin{eqnarray} \label{eq:29}
	\begin{gathered}
	M \ddot{{q}}+M \widehat{M}_{c}^{-1} \widehat{C}_{c} \dot{{q}}+M \widehat{M}_{c}^{-1}(J+R) \dot{{q}}+M \widehat{M}_{c}^{-1} \frac{\partial V_{c}}{\partial q} \\
	=-\left(G^{\perp}\right)^{T} G^{\perp}\left[C-M \widehat{M}_{c}^{-1} \widehat{C}_{c}-M \widehat{M}_{c}^{-1}(J+R)\right] \dot{{q}}-\left(G^{\perp}\right)^{T} G^{\perp}\left[\frac{\partial V}{\partial q}-M \widehat{M}_{c}^{-1} \frac{\partial V_{c}}{\partial q}\right]+G {u_{d}}
	\end{gathered}
\end{eqnarray}
and when the kinetic energy matching condition is hold, the motion equation of the closed loop system becomes as,
\begin{eqnarray} \label{eq:30}
	\widehat{M}_{c} \ddot{{q}}+\widehat{C}_{c} \dot{{q}}+(J+R) \dot{{q}}+\frac{\partial V_{c}}{\partial q}+\widehat{M}_{c} M^{-1}\left(G^{\perp}\right)^{T} G^{\perp}\left[\frac{\partial V}{\partial q}-M \widehat{M}_{c}^{-1} \frac{\partial V_{c}}{\partial q}\right]-\widehat{M}_{c} M^{-1} G {u_{d}}=0
\end{eqnarray}
To examine the stability of system given with (\ref{eq:30}), let us define a candidate Lyapunov function,
\begin{eqnarray} \label{eq:31}
	H_{c}({q}, \dot{{q}})=\frac{1}{2} \dot{{q}}^{T} \widehat{M}_{c}(q) \dot{{q}}+V_{c}(q)
\end{eqnarray}
The first derivative of this function along the trajectory is obtained from (\ref{eq:30}) and (\ref{eq:31}) as,
\begin{eqnarray} \label{eq:32}
	\dot{H}_{c}({q}, \dot{{q}})=-\dot{{q}}^{T} R \dot{{q}}-\dot{{q}}^{T} \widehat{M}_{c} M^{-1}\left(G^{\perp}\right)^{T} G^{\perp}\left[\frac{\partial V}{\partial q}-M \widehat{M}_{c}^{-1} \frac{\partial V_{c}}{\partial q}\right]-\dot{{q}}^{T} \widehat{M}_{c} M^{-1} G K_{v} G^{T} M^{-1} \widehat{M}_{c} \dot{{q}}
\end{eqnarray}
Let us define $\overline{\epsilon}(q)$ corresponded as the error of potential energy matching condition given in (\ref{eq:9}), and let define $\hat{\epsilon}(q), \epsilon(q, \dot{{q}})$, as follows,
\begin{eqnarray} \label{eq:33}
	\overline{\epsilon}({q})=G^{\perp}\left[\frac{\partial V}{\partial {q}}-M \widehat{M}_{c}^{-1} \frac{\partial V_{c}}{\partial {q}}\right], \qquad \hat{\epsilon}({q})=\widehat{M}_{c} M^{-1}\left(G^{\perp}\right)^{T} \overline{\epsilon}({q}), \qquad \epsilon({q}, \dot{{q}})=-\dot{{q}}^{T} \hat{\epsilon}({q})
\end{eqnarray}
and rewrite (\ref{eq:32}) as follows,
\begin{eqnarray} \label{eq:34}
	\dot{H}_{c}({q}, \dot{{q}})=-\dot{{q}}^{T} R \dot{{q}}-\dot{{q}}^{T} \widehat{M}_{c} M^{-1} G K_{v} G^{T} M^{-1} \widehat{M}_{c} \dot{{q}}+\epsilon({q}, \dot{{q}})
\end{eqnarray}
Let us define a function $P({q}, \dot{{q}})$ such that the first derivative is determined as follows,
\begin{eqnarray} \label{eq:35}
	\dot{P}({q}, \dot{{q}})=-\dot{{q}}^{T} R \dot{{q}}-\dot{{q}}^{T} \widehat{M}_{c} M^{-1} G K_{v} G^{T} M^{-1} \widehat{M}_{c} \dot{{q}}
\end{eqnarray}
and the relation (\ref{eq:32}) can be written as follows,
\begin{eqnarray} \label{eq:36}
	\dot{H}_{c}({q}, \dot{{q}})=\dot{P}({q}, \dot{{q}})+\epsilon({q}, \dot{{q}})
\end{eqnarray}
It is easily recognized that if potential energy matching condition could be exactly satisfied for $\widehat{M}_{c}({q})$ and $V_c({q})$, the first derivative of Lyapunov function given in (\ref{eq:31}) would be $\dot{P}({q}, \dot{{q}}) \leq 0$, therefore  the stability of the system is guaranteed. Since the $\dot{P}({q}, \dot{{q}}) \leq 0$, from (\ref{eq:36}) one can state that if $\epsilon({q}, \dot{{q}}) \leq 0$ along the trajectory, then the closed loop system is stable according to La Salle Theorem, whereas, if $\epsilon({q}, \dot{{q}})>0$ for some $({q}, \dot{{q}})$,  then the closed loop system stability cannot be determined forthrightly. So we will use the theorem given in Butz\cite{Butz1969}, Heinen\cite{Heinen1970} namely Lagrange stability, and some results on non-monotonic Lyapunov functions presented in Ahmadi\cite{Ahmadi2008,Ahmadi2014}. 
\begin{theorem} \label{thm1} \cite{Butz1969,Heinen1970}
	Suppose $f(x)$ is twice continuously differentiable, $V(x)$ is a real-valued three times continuously differentiable function defined on $R^{n}$ and $V(x) \rightarrow+\infty$ as $\|x\| \rightarrow+\infty$. Further suppose $\Omega$ is a bounded set in $R^n$ and $\widetilde{\Omega}$ its complement. Then system $\dot{x}=f(x)$ is Lagrange stable if, for some constants $\alpha_{1} \geq 0$ and $\alpha_{2} \geq 0$,
	\begin{eqnarray} \label{eq:37}
		\alpha_{2} \ddot{V}({x})+\alpha_{1} \ddot{V}({x})+\dot{V}({x})<0
	\end{eqnarray}
	for all $x \in \widetilde{\Omega}$.
	\hfill\ensuremath{\square}
\end{theorem}
As a consequence of above theorem, it is possible to determine the boundary of $\|\epsilon({q}, \dot{{q}})\|$ which guarantied the Lagrange stability condition given in (\ref{eq:37}). To achieve this manner it is enough to show that there is $\alpha_1>0$ and $\alpha_2>0$ exists which holds, 
\begin{eqnarray} \label{eq:38}
	\alpha_{2} \ddot{P}({q}, \dot{{q}})+\alpha_{1} \ddot{P}({q}, \dot{{q}})+\dot{P}({q}, \dot{{q}})+\alpha_{2} \ddot{e}({q}, \dot{{q}})+\alpha_{1} \dot{\epsilon}({q}, \dot{{q}})+\epsilon({q}, \dot{{q}})<0
\end{eqnarray}
There must be the amount of sub-regions $r$ that makes it possible for $\alpha_1\geq0$ and $\alpha_2\geq0$, such that satisfies (\ref{eq:38}), namely Lagrange stability of the system has been guaranteed. Existence such sub-regions will be clarified through examples.  
Furthermore, it might be useful to give the following version of above theorem.
\begin{theorem} \label{thm2} \cite{Ahmadi2008}
	Consider the continuous time system $\dot{x}=f(x)$. If there exists scalars $\alpha_1\geq0$ and $\alpha_2\geq0$, and three times differentiable Lyapunov function $V(x)$, such that 
	\begin{eqnarray} \label{eq:39}
		\alpha_{2} \ddot{V}({x})+\alpha_{1} \ddot{V}({x})+\dot{V}({x})<0
	\end{eqnarray}
	for all $x\neq0$, then for any $x(0)$  $V(x(t)) \rightarrow 0$ as $t \rightarrow \infty$ and the origin of the system $\dot{x}=f(x)$ is globally asymptotic stable. 
	\hfill\ensuremath{\square}
\end{theorem}
Moreover, Ahmadi\cite{Ahmadi2008} shows that existence of a monotonically decreasing Lyapunov function along the trajectories of dynamical system is not required for stability, instead converging to zero in limit is sufficient to guarantee stability.As a consequence of these results, the following corollary might be given under the condition of ${\left\lVert\overline{\epsilon}({q})\right\rVert}<\infty$ and so $\|\hat{\epsilon}(q)\|<\infty$.
\begin{corollary} \label{crl1}
	Let define,  
	$$ E(q) \triangleq \int \hat{\epsilon}^{T}(q) d q $$
	then the below relation is obtained for the term $\dot{q}^{T} \hat{\epsilon}(q)$ in (\ref{eq:36}), as follows,
	$$\hat{\epsilon}^{T}(q)=\left(\frac{\partial E(q)}{\partial q}\right)^{T}$$
	$$\dot{q}^{T} \hat{\epsilon}(q)=\frac{d E(q)}{d t}=\dot{q}^{T} \frac{\partial E(q)}{\partial q}$$
	it is clear that $H_{c}(q, \dot{q})$ can be written from (\ref{eq:36}) using (\ref{eq:33}) as follows,
	$$H_{c}(q, \dot{q})=P(q, \dot{q})-\int \dot{q}^{T} \hat{\epsilon}(q) d t$$
	after the following algebraic manipulations, 
	$$\int \dot{q}^{T} \hat{\epsilon}(q) d t=\int \frac{d E(q)}{d t} d t=\int d(E(q))=E(q)$$
	$$ \int_{0}^{t} \dot{q}^{T} \hat{\epsilon}(q) d t=\int_{q_{0}}^{q_{t}} \hat{\epsilon}^{T}(q) d q=E\left(q_{t}\right)-E\left(q_{0}\right) \leq 2\|E\| \quad \forall t$$
	$H_{c}(q, \dot{q})$ can be obtained as follows
	$$H_{c}(q, \dot{q})=P(q, \dot{q})-E(q)$$
	finally we stated that,
	$$H_{c}(q, \dot{q}) \leq P(q, \dot{q})-2\|E\|$$
	Since $\lim _{t \rightarrow \infty} P(q, \dot{q}) \rightarrow 0$, it is easily seen that $\lim _{t \rightarrow \infty} H_{c}(q, \dot{q}) \rightarrow 0$, as long as the term $\|E\|<\infty$. As consequent of the interpretation of Theorem 2, we can say that there exists a choice of the number of $r$ which allows the term $\|E\|$ to be finite, hence the control rule in (\ref{eq:22}) and (\ref{eq:24}) is stabilized the system given in (\ref{eq:1}) or equivalently (\ref{eq:4}), at the desired equilibrium point with dissipation gain $K_v$. \hfill\ensuremath{\square}
\end{corollary}
\section{cart and pendulum}\label{sec5}

In this section, we apply the preceding design methodology to the problem of stabilizing cart and pendulum shown in Fig. \ref{figCartPendulum}. We show that the method introduced in this paper offers a new method to solve matching conditions which provides closed loop stability.
\begin{figure}[ht]
\centerline{\includegraphics[width=342pt,height=9pc,keepaspectratio]{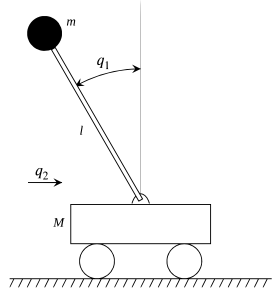}}
\caption{Cart and Pendulum System.\label{figCartPendulum}}
\end{figure}
The dynamic equation of the cart and pendulum is given in van der Schaft\cite{Schaft1996}:
\begin{eqnarray} \label{eq:41}
	\begin{gathered}
	{q}=\left[\begin{array}{l}{q_{1}} \\ {q_{2}}\end{array}\right], \quad M({q})=\left[\begin{array}{cc}{1} & {b \cos q_{1}} \\ {b \cos q_{1}} & {c}\end{array}\right], \quad V=a \cos q_{1} \\
	a=\frac{g}{l}, \qquad b=\frac{1}{l}, \qquad c=\frac{m+M}{l^{2} m}
	\end{gathered}
\end{eqnarray}
$m$ is the mass of the pendulum, $M$ is mass of cart, $l$ is the length of the pendulum and $g$ is the gravity. The position of the cart is the equilibrium point of the system which stabilized for  $q_1^*=0$ and an arbitrary $q_2^*$.
The solution procedure will be given as follows step by step. 
\setcounter{step}{0}
\begin{step}
	Let us define $r=13$ sub-regions in the configuration space of the EL system;
	\begin{eqnarray} \label{eq:42}
		\mathbb{S}_{i} \triangleq\left\{{q} | h_{i}({q}) \geq h_{l}({q}), \quad l=1,2, \ldots, 13, \quad i \neq l, \quad-\frac{\pi}{2}<q<\frac{\pi}{2}\right\}	
	\end{eqnarray}	
	where $h_{i}({q})^{\prime}$s are chosen as $h_{i}({q})=e^{-\left(\epsilon_{i l}\left\|{q}-{q}_{i}\right\|\right)^{2}}$. Note that the ${q}_{i}^{\prime}$s are the value of generalized coordinates which correspond $h_{i}\left.({q})\right|_{q=q_{i}}=1$
\end{step}
\begin{step}
	Write the potential energy matching condition at the centers of these sub-regions, 
	
	Let us define a matrix to solve condition given in (\ref{eq:14})
	\begin{eqnarray} \label{eq:43}
		S_{i}=M_{i} M_{c i}^{-1}
	\end{eqnarray}	
	Potential energy matching condition of the cart and pendulum system, where parameters is given as follows (\ref{eq:41}),
	\begin{eqnarray} \label{eq:44}
		\begin{gathered}
		G^{\perp}\left[\left[\begin{array}{cc}{m g l \sin q_{1}} \\ {0} & {0}\end{array}\right]-S_{i}\left[\begin{array}{c}{\partial V_{c} / \partial q_{1}} \\ {\partial V_{c} / \partial q_{2}}\end{array}\right]\right]=0 \\
		m g l \sin q_{1}=S_{i}(1,1) \frac{\partial V_{c}}{\partial q_{1}}+S_{i}(1,2) \frac{\partial V_{c}}{\partial q_{2}}
		\end{gathered}
	\end{eqnarray}
	General solution of this PDE as, 
	\begin{eqnarray} \label{eq:45}
		\begin{gathered}
		V_{c}({q})=\frac{a}{S_{i}(1,1)} \cos \left(q_{1}\right)+\Phi(z({q})) \\
		z({q})=\left(q_{2}-q_{2}^{*}\right)-\frac{S_{i}(1,2)}{S_{i}(1,1)} q_{1}
		\end{gathered}
	\end{eqnarray}
	where $\Phi(z(q))$ is an arbitrary differentiable function which satisfies the condition $\nabla_{q} \Phi(z(0))=0 .$  According to Remark \ref{remark2} to establish $V_{\text { hess }}\left(q^{*}\right)>0$, it is taken as $\Phi(z)=\left(K_{p} / 2\right) z^{2}$ where $K_{p}>0$
	\begin{eqnarray} \label{eq:46}
		V_{c}({q})=K_{p}\left(\frac{q_{2}^{2}}{2}-q_{2} q_{c 2}+\frac{q_{c 2}^{2}}{2}+\frac{q_{1}\left(-q_{2}+q_{c 2}\right) s_{12}}{s_{11}}+\frac{q_{1}^{2} s_{12}^{2}}{2 s_{11}^{2}}\right)+\frac{a \cos q 1}{s_{11}}
	\end{eqnarray}
	The gradient and Hessian of the $V_c ({q})$ is obtained as follows:
	\begin{eqnarray} \label{eq:47}
		\nabla_{{q}} V_{c}({q})=\left[\begin{array}{c}{\frac{K_{p} s_{12}\left(-q_{2}+q_{2}^{*}+\frac{q_{1} s_{12}}{s_{11}}\right)-a \sin q_{1}}{s_{11}}} \\ {-K_{p}\left(-q_{2}+q_{2}^{*}+\frac{q_{1} s_{12}}{s_{11}}\right)}\end{array}\right]
	\end{eqnarray}
	\begin{eqnarray} \label{eq:48}
		V_{c_{hess}}^{{q}}({q})=\left[\begin{array}{cc}{\frac{\left(K_{p} s_{12}^{2}-a s_{11} \cos q_{1}\right)}{s_{11}^{2}}} & {-\frac{K_{p} s_{i}(1,2)}{S_{i}(1,1)}} \\ {-\frac{K_{p} S_{i}(1,2)}{S_{i}(1,1)}} & {K_{p}}\end{array}\right]
	\end{eqnarray}
	To establish $V_{c_{hess}}^{q}(q)>0,$ next conditions are found as:
	\begin{eqnarray} \label{eq:49}
		S_{i}(1,1)<0 \quad and \quad K_{p}>0
	\end{eqnarray}
\end{step}
\begin{step}
	Let us define an generalized inertia matrix of closed loop system in terms of $h_i({q})$s

	First we need to find additional conditions which holds $M_{ci}>0$. $M_{ci}$ matrices are solved $\forall i$ which are described respect to $S_i$ matrix elements. The matrices $M_i$s are given as:
	\begin{eqnarray} \label{eq:50}
		M_{i}=\left[\begin{array}{cc}{1} & {m_{i}} \\ {m_{i}} & {c}\end{array}\right], \qquad i=1,2, \dots, 11
	\end{eqnarray}
	where $m_{i}=b \cos q_{1}^{i}$ and $q_{1}^{i}=\frac{\pi}{2}-i\frac{\pi}{12}$, thus values of $q_1=\pm\frac{\pi}{2}$ have not been evaulated because they cause singularity, then the general form of $M_{c i}$s are obtained as follows,
	\begin{eqnarray} \label{eq:51}
		M_{c i}=\left[\begin{array}{cc}{\frac{-m_{i} S_{i}(1,2)+S_{i}(2,2)}{-S_{i}(1,2) S_{i}(2,1)+S_{i}(1,1) S_{i}(2,2)}} & {\frac{-6 S_{i}(1,2)+m_{i} S_{i}(2,2)}{-S_{i}(1,2) S_{i}(2,1)+S_{i}(1,1) S_{i}(2,2)}} \\ {\frac{m_{i} S_{i}(1,1)-S_{i}(2,1)}{-S_{i}(1,2) S_{i}(2,1)+S_{i}(1,1) S_{i}(2,2)}} & {\frac{6 S_{i}(1,1)-m_{i} S_{i}(2,1)}{-S_{i}(1,2) S_{i}(2,1)+S_{i}(1,1) S_{i}(2,2)}}\end{array}\right], \quad i=1,2, \ldots, 11
	\end{eqnarray}
	To ensure that $M_{ci}>0$ $\forall i$, the following conditions, which are derived by using Mathematica, must be satisfied,
	\begin{eqnarray} \label{eq:52}
		\begin{gathered}
			S_{i}(1,2)>0, \quad-\frac{S_{i}(1,1)}{S_{i}(1,2)}<\min \left\{m_{i}\right\}, \\ 
			\quad S_{i}(2,1)<\frac{m_{1} S_{i}^{2}(1,1)+6 S_{1}(1,1) S_{i}(1,2)}{S_{i}(1,1)+m_{i} S_{i}(1,2)}, \quad S_{i}(2,2)=\frac{m_{i} S_{i}(1,1)+6 S_{i}(1,2)-S_{i}(2,1)}{m_{i}}
		\end{gathered}
	\end{eqnarray}
	in which $S_{i}(1,1)<0$. Let us define some constant scalars as $\alpha_{1}=-\frac{S_{i}(1,1)}{S_{i}(1,2)}$, $\alpha_{2}=S_{i}(1,2)$. From the second condition of (\ref{eq:52}), it is easily realized that $0<\alpha_{1}<\min \left\{m_{i}\right\}$ and $\alpha_2>0$  and for the third condition of (\ref{eq:52}) let us define the scalar $\beta>0$. Finally the conditions which are guaranteed $M_{ci}>0$ $\forall i$ can be given as follows,
	\begin{eqnarray} \label{eq:53}
		\begin{gathered}
			S_{i}(1,1)=-\alpha_{1} \alpha_{2}, \quad S_{i}(1,2)=\alpha_{2}, \quad \\
			S_{i}(2,1)=\frac{\alpha_{1} \alpha_{2}\left(6-\alpha_{1} m_{i}\right)}{\alpha_{1}-m_{i}}-\beta, \quad S_{i}(2,2)=\frac{\left(6 \alpha_{2}+\beta\right) m_{i}-\alpha_{1}\left(\beta+\alpha_{2} m_{i}^{2}\right)}{m_{i}\left(-\alpha_{1}+m_{i}\right)}
		\end{gathered}
	\end{eqnarray}
	The matrices $M_{ci}$s can be rewritten in the terms of $\alpha_1$, $\alpha_2$, $\beta$ as follows, 
	\begin{eqnarray} \label{eq:54}
		M_{c i}=\left[\begin{array}{cc}{\frac{-\alpha_{1} \beta+m_{i}\left(\beta-\alpha_{2}\left(-6+m_{i}^{2}\right)\right)}{\alpha_{2} \beta\left(\alpha_{1}-m_{i}\right)^{2}}} & {\frac{m_{i}\left(\beta m_{i}-\alpha_{1}\left(\beta+\alpha_{2}\left(-6+m_{i}^{2}\right)\right)\right)}{\alpha_{2} \beta\left(\alpha_{1}-m_{i}\right)^{2}}} \\ {\frac{m_{i}\left(\beta m_{i}-\alpha_{1}\left(\beta+\alpha_{2}\left(-6+m_{i}^{2}\right)\right)\right)}{\alpha_{2} \beta\left(\alpha_{1}-m_{i}\right)^{2}}} & {-\frac{m_{i}\left(\alpha_{1} \beta m_{i}-\beta m_{i}^{2}+\alpha_{1}^{2} \alpha_{2}\left(-6+m_{i}^{2}\right)\right)}{\alpha_{2} \beta\left(\alpha_{1}-m_{i}\right)^{2}}}\end{array}\right], \quad i=1,2, \ldots, 11
	\end{eqnarray}
	in which $\beta>0$, $\alpha_2>0$, $0<\alpha_{1}<\min\{m_{i}\}$ $\forall i$. Since the existence of $\alpha_1$ is guaranteed as long as $\min\{m_{i}\}>0$, which is always positive, so it is better to define a new parameter $0<\gamma<1$ such that $\gamma=1-\frac{\alpha_{1}}{\min\{m_{i}\}}$.
\end{step}
\begin{step}
	Construct an approximate generalized inertia matrix of closed loop system: 
	\begin{eqnarray} \label{eq:55}
		\widehat{M}_{c}({q})=\sum_{i}\left( e^{-\left(\epsilon_{i}\left(q_{i}-q\right)\right)^{2}} \widetilde{M}_{c i}+M_{c b i}\right)
	\end{eqnarray}
	Using Matlab neural network toolbox, the parameters $\epsilon_{i}$s, $\widetilde{M}_{c i}s$ and $M_{c b i}$s have been found explained in Remark \ref{remark3} which is satisfied (\ref{eq:16}). The approximation error of potential energy matching condition $\overline{\epsilon}({q})$, is given in (\ref{eq:33}), is plotted in Figure \ref{figCartPendulumPotError} in the case of 13 amount of sub-regions in the range of $\left\{-1.309<q_{1}<1.309\right\}$
	\begin{figure}[ht]
	\centerline{\includegraphics[width=450pt,height=18pc,keepaspectratio]{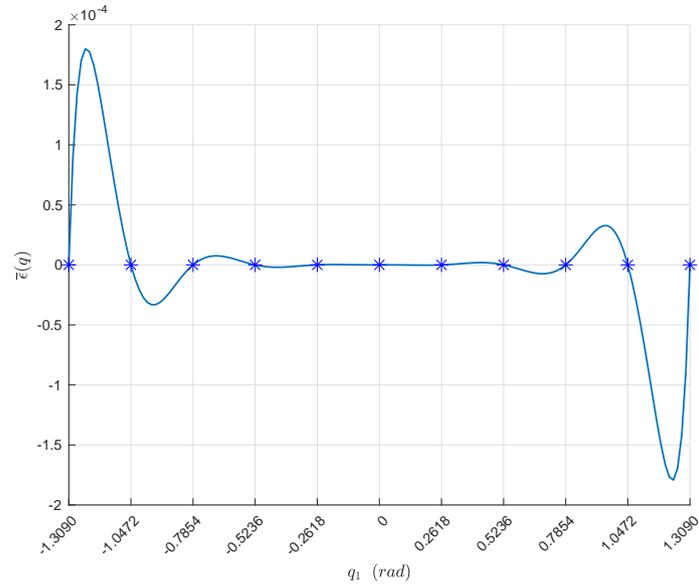}}
	\caption{Potential energy matching condition error, $\overline{\epsilon}({q})$ respect to $q_1$.\label{figCartPendulumPotError}}
	\end{figure}
\end{step}
\begin{step}
	Find some $J({q}, \dot{{q}})=-J^{T}({q}, \dot{{q}})$ and $R({q}, \dot{{q}})=R^{T}({q}, \dot{{q}}) \geq 0$ matrices such that hold the kinetic energy matching condition given in (\ref{eq:17}). 
	$J({q}, \dot{{q}})$ can be chosen as in (\ref{eq:20}) and $R({q}, \dot{{q}})$ can be chosen as follows for Cart and Pendulum system:

	\begin{eqnarray} \label{eq:56}
		R(q, \dot{q})=\left[\begin{array}{cc}{\frac{\varphi \Sigma(1,2)^{2}}{\Sigma(1,1)^{2}}} & {-\frac{\varphi \Sigma(1,2)}{\Sigma(1,1)}+\frac{\Sigma(1,1) \Pi(1,1)}{\Sigma(1,2)}+\Pi(1,2)} \\ {-\frac{\varphi \Sigma(1,2)}{\Sigma(1,1)}+\frac{\Sigma(1,1) \Pi(1,1)}{\Sigma(1,2)}+\Pi(1,2)} & {\varphi-\frac{\Sigma(1,1)^{2} \Pi(1,1)}{\Sigma(1,2)^{2}}+\Pi(2,2)}\end{array}\right]
	\end{eqnarray}
	in which $\Sigma ({q}, \dot{{q}})$ and $\Pi ({q}, \dot{{q}})$ defined as follows:
	\begin{eqnarray} \label{eq:57}
		\begin{gathered}
		\Sigma=G^{\perp} M M_{c}^{-1} \\
		\Pi=\frac{1}{2}\left(M_{c} M^{-1} C-C_c+\left(M_{c} M^{-1} C-C_c\right)^{T}\right)
		\end{gathered}
	\end{eqnarray}
	where $R({q}, \dot{{q}}) \geq 0$ for a constant $\varphi \geq 0$.
\end{step}
\begin{step}
	Construct energy shaping control rule.
	The control input, obtained from matching conditions is defined in (\ref{eq:22}).  The controller design parameters are chosen as, $K_{p}=\frac{1}{100}, \gamma=0.5, \alpha_{2}=20, \beta=30, \varphi=0.002$.
	\begin{eqnarray} \label{eq:58}
		u_{c}=\left(G^{T} G\right)^{-1} G^{T}\left\{\left[C(q, \dot{q})-M \widehat{M}_{c}^{-1}\left(C_{c}(q, \dot{q})+(J+R)\right)\right] \dot{q}+\left[\frac{\partial V}{\partial q}-M \widehat{M}_{c}^{-1} \frac{\partial V_{c}}{\partial q}\right]\right\}
	\end{eqnarray}
\end{step}
\begin{step}
	Closed loop system potential energy function is plotted in Figure \ref{figVc} for control parameters given below.
	\begin{figure}[ht]
	\centerline{\includegraphics[width=342pt,height=20pc,keepaspectratio]{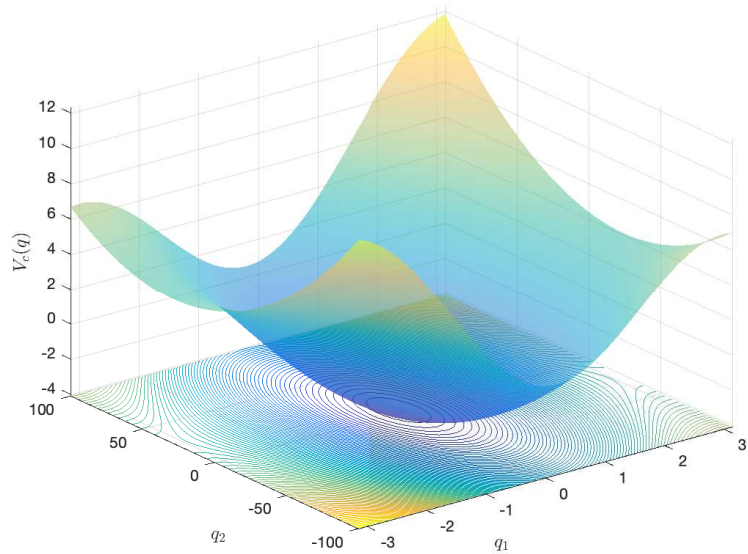}}
	\caption{Closed loop potatial energy function, $V_c(q)$.\label{figVc}}
	\end{figure}
	To provide system asymptotic stability, dissipation added as follows:
	\begin{eqnarray} \label{eq:59}
		u_{d}=-k_{v} G^{T} M^{-1} \widehat{M}_{c} \dot{q}, \qquad k_{v}=700
	\end{eqnarray}
	The results are presented in Figure \ref{figCartPendulumPotKv2} which illustrate time domain responses and $q_1$ and $q_2$ path via $V_c({q})$ is illustrated in Figures \ref{figCartPendulum2d},\ref{figCartPendulum3d}. Finally $\dot{H}({q}, \dot{{q}})$ given in (\ref{eq:30}) and ${H}({q}, \dot{{q}})$ are illustrated in Figure \ref{figHq}:
	\begin{figure}[ht]
	\centerline{\includegraphics[width=342pt,height=20pc,keepaspectratio]{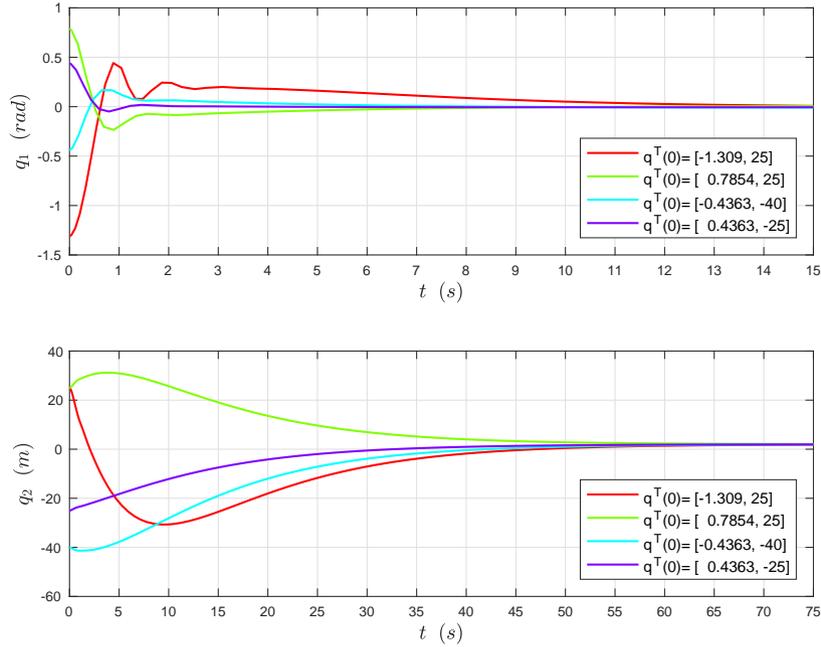}}
	\caption{System states for different initial conditions.\label{figCartPendulumPotKv2}}
	\end{figure}
	\begin{figure}[ht]
	\centerline{\includegraphics[width=342pt,height=20pc,keepaspectratio]{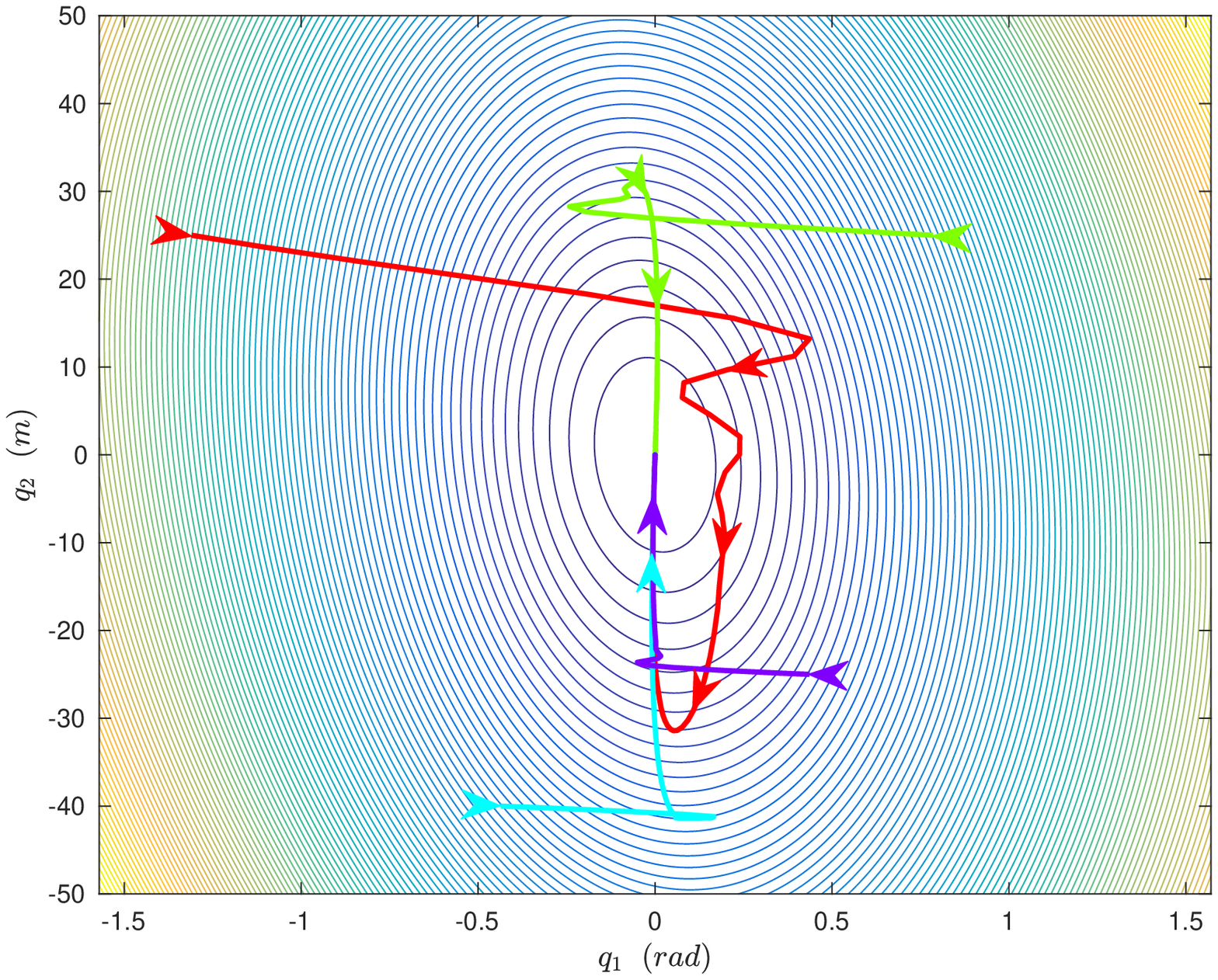}}
	\caption{Trajectory for different initial conditions.\label{figCartPendulum2d}}
	\end{figure}
	\begin{figure}[ht]
	\centerline{\includegraphics[width=342pt,height=20pc,keepaspectratio]{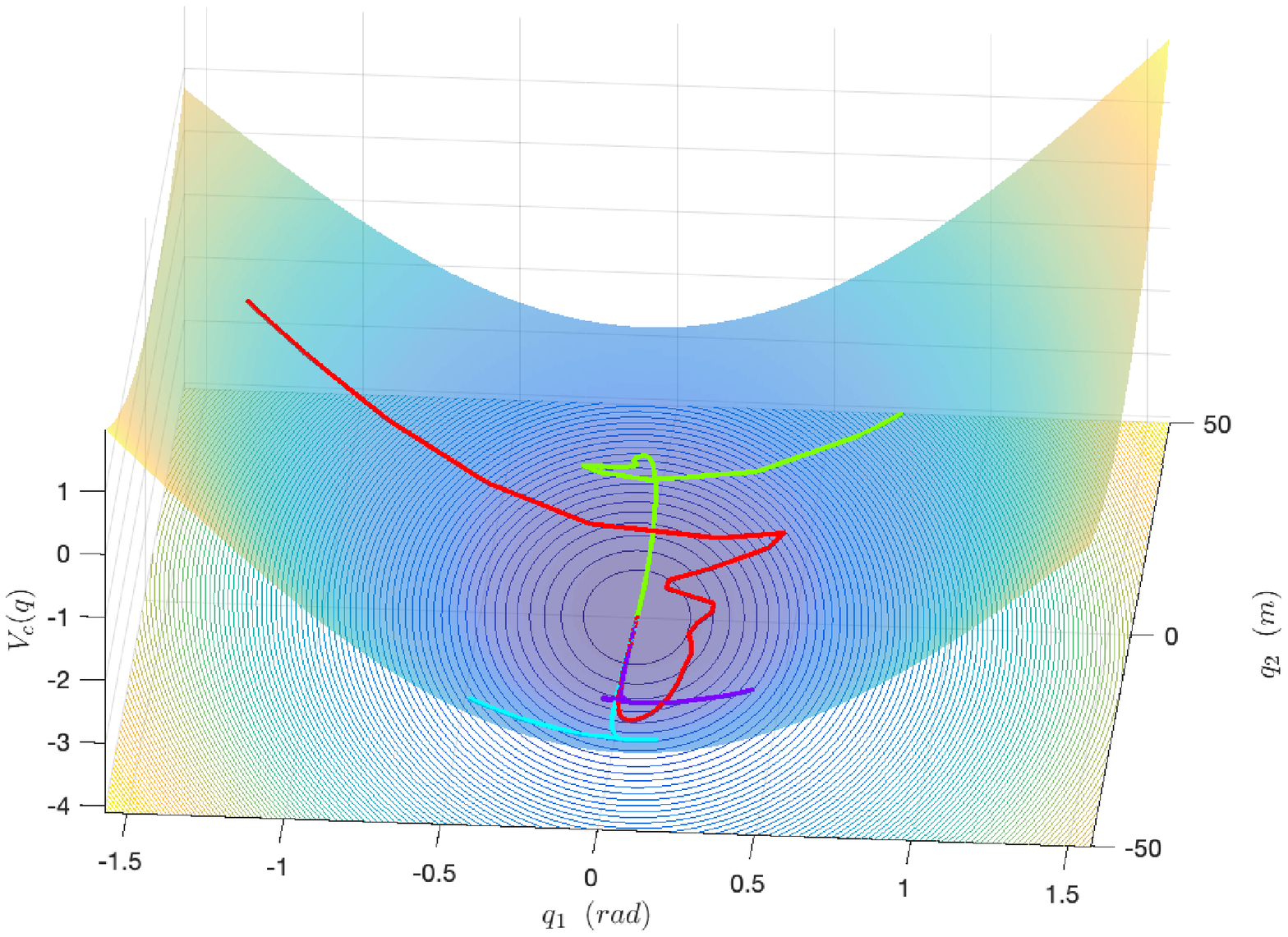}}
	\caption{Trajectory for different initial conditions.\label{figCartPendulum3d}}
	\end{figure}
	\begin{figure}[ht]
	\centerline{\includegraphics[width=342pt,height=20pc,keepaspectratio]{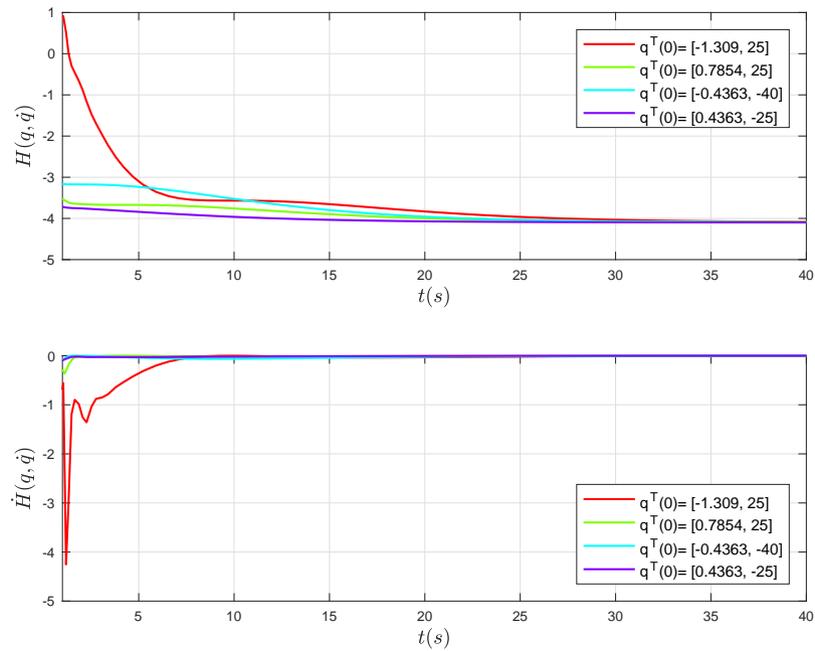}}
	\caption{$H(q,\dot{q})$ and $\dot{H}(q,\dot{q})$ for different initial conditions\label{figHq}}
	\end{figure}
\end{step}

\clearpage

\section{Conclusions}\label{sec6}

In this paper, the stability problem of underactuated EL system is considered. The standard method for stabilizing of EL systems is controlled Lagrangian method. In this method, the constructing of the control law requires to solve a set of nonlinear nonhomogeneous partial differential equation. In this study, we proposed a method to obtain an approximate solution of these PDEs based on an approximate model of the system. Furthermore, the stability analyzes of the closed loop system which is controlled by the control rule using found the proposed method here is done using non-monotonic Lyapunov functions

\bibliography{yildiz_goren-sumer}%

\end{document}